\documentclass[12pt,twoside]{article}
\usepackage{amsmath,amssymb,euscript}
\usepackage[numbers]{natbib}
\usepackage{hyperref}

\usepackage[margin=1.2in]{geometry}
\usepackage{amssymb,mathrsfs}
\usepackage{enumerate}

\usepackage{amsmath}
\usepackage{amsthm}
\usepackage{amssymb}
\usepackage{graphicx}

\usepackage{epsf}
\usepackage{epic}
\usepackage{makeidx}
\usepackage{psfrag}

\def\opA{{\mathbb A}}
\def\opH{{\mathbb H}}
\def\opN{{\mathbb N}}

\def\R{{\mathbb R}}
\def\Z{{\mathbb Z}}
\def\T{{\mathbb T}}

\def\cA{{\mathcal A}}
\def\cB{{\mathcal B}}

\def\cE{{\mathcal E}}

\def\cH{{\mathcal H}}

\def\cL{{\mathcal L}}

\def\cN{{\mathcal N}}

\def\cS{{\mathcal S}}

\def\cX{{\mathcal X}}

\def\g{{\gamma}}
\def\gam{{\gamma}}
\newcommand{\al}{\alpha}
\newcommand{\s}{\sigma}

\def\om{\omega}

\def\qf{\om^q}

\def\states{{\mathfrak S}}
\def\qfstates{{\mathfrak Q}}

\def\HHFB{\opH_{hfb}}

\def\Hd{{H^1}}
\def\Hsdd{{H^1_s}}
\def\cHLd{{\cH^1}}
\def\XL{{X^1}}

\newcommand{\p}{\partial}

\newcommand{\ra}{\rightarrow}

\def\tr{{\rm Tr}}

\def\1{{\bf 1}}

\newcommand{\ls}{\lesssim}

\newcommand{\DETAILS}[1]{}

\def\eqnn{\begin{eqnarray*}}
\def\eeqnn{\end{eqnarray*}}
\def\eqn{\begin{eqnarray}}
\def\eeqn{\end{eqnarray}}

\def\endprf{\end{proof}} 

\newtheorem{theorem}{Theorem}[section]

\newtheorem{lemma}[theorem]{Lemma}
\newtheorem{corollary}[theorem]{Corollary}
\newtheorem{remark}[theorem]{Remark}

\title
{On the Hartree-Fock-Bogoliubov equations}

\author{V. Bach\footnote{Institut fuer Analysis und Algebra Carl-Friedrich-Gauss-Fakult\"at, Technische Universit\"at Braunschweig, 38106 Braunschweig, Germany, v.bach@tu-bs.de}, S.  Breteaux\footnote{Institut \'Elie Cartan de Lorraine, Universit\'e de Lorraine - Site de Metz, 57000, 
Metz, France, sebastien.breteaux@univ-lorraine.fr}, Th. Chen\footnote{Department of Mathematics, University of Texas at Austin, Austin TX 78712, USA, tc@math.utexas.edu}, J. Fr\"ohlich\footnote{Institut f{\"u}r Theoretische Physik, ETH H{\"o}nggerberg, CH-8093 Z{\"u}rich, Switzerland, juerg@phys.ethz.ch}, I. M. Sigal\footnote{Department of Mathematics, University of Toronto, Toronto, ON M5S 2E4, Canada, im.sigal@utoronto.ca}}

\begin{document}


\maketitle


\begin{abstract}
We review some results of our paper  \cite{BachBreteauxChenFroehlichSigal2016v2} on the ``nonlinear quasifree approximation'' to the many-body Schr\"odinger dynamics of Bose gases. In that paper, we derive, with the help of this approximation,  the time-dependent Hartree-Fock-Bogoliubov (HFB) equations, providing an approximate description of the dynamics of quantum fluctuations around a Bose-Einstein condensate and study properties of these equations. 
\end{abstract} 


\section{Introduction}
\label{sec-intro-1}
The Schr\"odinger equation is used to describe aspects of the dynamics of quantum systems as diverse as atoms, solids and stars. Although it has a very 
compact appearance and is easy to write down, understanding its solutions is fiendishly complicated as soon as more than two particles are involved. Hence, in order to be able to use quantum theory to derive interesting predictions concerning the behavior of physical systems, it is crucial to develop approximation techniques. The most powerful of these yield effective equations that provide fairly accurate descriptions of dynamical physical phenomena, yet are rather simple to handle.

One such technique that works especially well for large systems of identical particles is the self-consistent one
-body approximation yielding equations known as the Hartree- and Hartree-Fock equations, which are used to describe systems of many interacting bosons and fermions, respectively, at zero as well as at positive densities and temperatures. Their generalizations, the  Hartree-Fock-Bogolubov (HFB) and Bogolubov-de Gennes (BdG) equations, were developed to study properties of quantum fluids, such as Bose-Einstein condensation and superfluidity, for bosons, and superconductivity, for charged fermions forming Cooper pairs.

In  \cite{BachBreteauxChenFroehlichSigal2016v2}, we have proposed a simple algorithm for deriving such effective equations, 
 the ``quasifree approximation''.\footnote{This approximation is called ``quaisfree reduction'' in  \cite{BachBreteauxChenFroehlichSigal2016v2}. 
But, as it turned out, the latter expression was already used - in the gauge-invariant context - for a different notion; see  below.} 
We have then applied it to derive the time-dependent extension of the \textit{Hartree-Fock-Bogoliubov (HFB) equations}. 
 Moreover, we have initiated a mathematical theory of these equations. 

In this note, we review the results presented in \cite{BachBreteauxChenFroehlichSigal2016v2}. 
First, we recall some features of the many-body problem. Subsequently, our main results are outlined and some proofs are sketched. 
In particular, we sketch the proof of the existence of solutions to the HFB equations given in the second version v2 of \cite{BachBreteauxChenFroehlichSigal2016v2} which assumes stronger hypotheses than those in the first version v1 of 
\cite{BachBreteauxChenFroehlichSigal2016v1}. In fact, there was an error in \cite{BachBreteauxChenFroehlichSigal2016v1} (kindly pointed out to us by J.~Sok) in the proof of one of the estimates needed in the proof of local well-posedness (see  \cite[Lemma~E.1(2)]{BachBreteauxChenFroehlichSigal2016v1}). 
In Section \ref{sec-k-est} we prove the required estimate but under a stronger condition on $v$ (see Lemma \ref{lem:k-est}).
A proof under weaker conditions (similar to those in \cite{BachBreteauxChenFroehlichSigal2016v1}) will be given elsewhere. 

\section{Quantum-mechanical many-body problem}\label{sec:many-body}
In quantum theory, 
  the time evolution of the quantum state, $\omega_t$, of a many-body system 
   is given by the 
von~Neumann-Landau equation 
\eqn\label{eq-vNeum-1} 
    i\partial_t\om_t(\opA) = \om_t([\opA,\opH]) \,,    
\eeqn
where $\opA$ is an arbitrary operator - an ``observable'' - belonging to the Weyl algebra, ${\frak W}$, over Schwartz space, 
$\cS(\R^d)$, and $\opH$ is the quantum Hamiltonian  
\eqn  \label{model-ham} 
    \opH := \int dx \; \psi^*(x)\,(h\psi)(x) 
    + \frac{1}{2}\int dx\int dy \; v(x-y)\,\psi^*(x) \psi^*(y)
    \psi(x) \psi(y) \,.
\eeqn
Here, $x$ and $y$ denote points in physical space $\mathbb{R}^{d}$, $h$ is the ``one-particle hamiltonian'' given by $h:= -\Delta +V(x)$, $\Delta$ is the Laplacian acting on 
$L^{2}(\mathbb{R}^{d})$, $V(x)$ is the potential of an external force acting on the particles, and $\psi^{*}(x)$ and $\psi(x)$ are the operator-valued distributions, called the 
creation- and annihilation operators, satisfying, for Bose systems, the canonical commutation relations (CCR). The algebra $\frak{W}$ is generated by exponentials of the selfadjoint operators $\psi(f)+\psi^{*}(\overline{f})$, where $f$ is an arbitrary test function in $\cS(\R^d)$. For details, see, e.g., \cite{BachBreteauxChenFroehlichSigal2016v2, BratteliRobinson-II-1996}.  
%

Let $W^{p, r}(\R^d)$ denote the standard Sobolev space over $\R^d$. We will require the following assumptions. \\
\begin{itemize}
\item[(i)] The external potential $V$ is infinitesimally bounded with respect to the Laplacian.
\item[(ii)] The pair potential $v$ is 
 infinitesimally $-\Delta$\,-\,bounded and is even, $v(x)=v(-x)$. 
\end{itemize}
Conditions (i) and (ii) imply that, for systems of finitely many particles, the Hamiltonian $\opH$ is well-defined on a dense domain in the bosonic Fock space $\mathcal{F}$ and self-adjoint on the domain of the operator 
$\opH_0 := \int dx \; \psi^*(x)(-\Delta)\psi(x)$; see \cite{BachBreteauxChenFroehlichSigal2016v2}. In Section \ref{sec-WP-HFB} we will use a stronger condition on $v$: 
\begin{itemize}
\item[(ii')] The pair potential $v$ satisfies  $v \in W^{p, 1}$ with $p>d$ and $v(x) = v(-x)$. \end{itemize}

\section{Quasifree approximation of the full dynamics}

Perhaps the simplest general set of states of a quantum many-body systems consists of the quasifree states; see, e.g., \cite{BratteliRobinson-II-1996}.\footnote{The notion of quasifree states was introduced in \cite{Robinson1965}; see 
\cite{BratteliRobinson-II-1996} and references therein.}
   Quasifree states generalize the Hartree- and Hartree-Fock states, as has been first realized and used in \cite{BachLiebSolovej1994}. They represent a non-abelian version of Gaussian measures. Our goal in this paper is to approximate the general many-body dynamics given in Eq. \eqref{eq-vNeum-1}  by a dynamics that leaves the set of quasifree states invariant. 

We denote the space of all states on the Weyl CCR algebra
${\frak W}$ by $\states$ and the subset of quasifree states by 
$\qfstates \subseteq \states$. (We will distinguish quasifree
states, $\qf \in \qfstates$, from general states, $\om \in \states$, by adding
a superscript ``$q$''.)

We propose to map the solution $\om_t$ of the von Neumann-Landau equation \eqref{eq-vNeum-1} with a quasifree initial condition $\omega_{t=0}=\qf_0 \in \qfstates$ to a family, $(\qf_t)_{t \geq 0} \in C^1\big(\R_0^+; \qfstates \big)$, of quasifree states satisfying the equation
\eqn\label{eq-vNeum-quasifree}
    i\partial_t\qf_t(\opA) = \qf_t([\opA,\opH]) \quad \text{ with initial condition } \quad \qf_{t=0}=\qf_0\,,   
\eeqn
for all observables $\opA$ which are \textit{linear} or \textit{quadratic} in the creation- and annihilation operators. 
We call the map $\omega_{t} \mapsto \omega^{q}_{t}$, as determined by Eq. \eqref{eq-vNeum-quasifree}, the \textit{nonlinear quasifree approximation} of equation \eqref{eq-vNeum-1}.\footnote{As was mentioned above, this map is called ``quasifree reduction'' in  \cite{BachBreteauxChenFroehlichSigal2016v2}. There is another natural map {\bf (cf. \cite{ArakiShiraishi1971}}) (see below) $q: \om \ra \qf$, defined  by $\mu(\qf)=\mu(\om)$, 
where $\mu$ is the map from states to truncated expectations of linear or quadratic as the operator with expressions in creation- and annihilation operators: \[\mu: \om\ra \om [\psi(x)],\ \om [\psi^{*}(y) \, \psi(x)] - \om [\psi^{*}(y)] \, \om [\psi(x)],\  \om [\psi(x) \, \psi(y)] - \om [\psi(y)] \, \om [\psi(x)].\]   
For {\it gauge-invariant} states, i.e., states $\omega$ with $\omega [\psi(x)]=0,\  \omega [\psi(x) \, \psi(y)]=0$); a related map is called ``quasifree reduction'' in \cite{OhyaPetz1993}. 
}
As shown below, this map determines the time-dependent generalization of the \textit{Hartree-Fock-Bogoliubov (HFB) equations}. 

 We emphasize that, in contrast to the von Neumann-Landau equation \eqref{eq-vNeum-1}, equation \eqref{eq-vNeum-quasifree} is \textit{non-linear}.
Equation \eqref{eq-vNeum-quasifree} turns out to be equivalent to the  self-consistent equation
\eqn\label{eq-qf-self-consist}
    i\partial_t \qf_t(\opA) = \qf_t([\opA,\opH_{\rm hfb}(\qf_t)]), \qquad \forall \mathbb{A}\in \frak{W},
\eeqn
where $\opH_{\rm hfb}(\qf)$ is an explicit {\it quadratic} Hamiltonian depending on the state $\omega^{q}$ (see \eqref{HHFB-def}, below). 
\DETAILS{\begin{theorem} \label{thm-Hquadr-dyn-1}
Under conditions (i) and (ii),  a quasifree state $\qf_t\in \cX_{T}^{ \rm qf}$ 
satisfies  \eqref{eq-vNeum-quasifree} if and only if it  satisfies the self-consistent equation 
\eqn\label{eq-omt-commut-3}
    i\partial_t \qf_t(\opA) = \qf_t ({[\opA, \HHFB(\qf_t)]}) \,.
\eeqn
\end{theorem}}

In subsequent work \cite{BenedikterSokSolovej2018}, the relation of the quasifree approximation (reduction) of \cite{BachBreteauxChenFroehlichSigal2016v2} to the Dirac-Frenkel principle used to derive the Hartree-Fock equations has been clarified; (see \cite{Dirac1930, Frenkel1934} for the original works and \cite{Lubich2008, Griesemer2017} for a recent review and an application).

\section{HFB equations and their properties}\label{sec:HFBeqs-prop} 


Recall that a quasifree state $\om^q$ is uniquely determined by the truncated expectations of linear and quadratic expressions in creation- and annihilation operators:
    \begin{equation} \label{omq-mom} 
    \begin{cases}
    \phi(x):=\om^q [\psi(x)],\\
    \gamma (x;y) :=   \om^q[\psi^{*}(y) \, \psi(x)] - \om^q [\psi^{*}(y)] \, \om^q [\psi(x)] ,\\
    \sigma (x,y):=   \om^q[\psi(x) \, \psi(y)] - \om^q [\psi(x)] \, \om^q [\psi(y)] \,.
\end{cases} 
\end{equation}

Let $\gamma$ and $ \sigma$ denote the operators on $L^2(\R^d)$ with integral kernels given by $\gamma (x,y)$ and $\sigma (x,y)$, respectively. It is obvious from definition~\eqref{omq-mom} that
\begin{equation} \label{gam-sig-cond}\gamma=\gamma^*\ge 0\ \text{ and }\ \sigma^*=\bar\sigma, \end{equation} where $\bar\sigma =C \sigma C$, and $C$ is complex conjugation.
More precisely, for any state $\om$, with the truncated expectations $(\phi, \gamma, \sigma)$ defined as in \eqref{omq-mom}, we have 
\begin{equation} \label{Gam-pos} 
\Gamma :=\begin{pmatrix}
\gamma & \sigma\\
 \bar{\sigma} & 1+\bar \gamma
\end{pmatrix}\geq 0\,.
\end{equation}
In the opposite direction, it was shown in \cite[Lemmata~3.2-3.5]{ArakiShiraishi1971} that, if \eqref{Gam-pos} holds then there is a quasi-free state having these $(\phi, \gamma, \sigma)$ as its truncated expectations. 
(The positivity condition on $\Gamma$ in \eqref{Gam-pos} can be expressed directly in terms of $\g$ and $\s$; see \cite{BachBreteauxKnoerrMenge2014}.)

The matrix operator in \eqref{Gam-pos} is called ``generalized one-particle density matrix''. 
 We will use \eqref{Gam-pos} in proving the global existence for the HFB equations (see \eqref{sig-est} of Section \ref{sec-WP-HFB}).
\DETAILS{ 
\eqref{omq-mom} implies that 
 \begin{equation} \label{Gam-pos}\Gamma=\begin{pmatrix}
\gamma & \sigma\\
 \bar{\sigma} & 1+\bar \gamma
\end{pmatrix}\geq 0\,.
\end{equation}
The converse is shown in  \cite{BachLiebSolovej1994} for fermion states and in \cite{Solovej2014} for pure boson ones, with $\tr\g<\infty$ in both cases. 
In the following, we always assume that \eqref{Gam-pos} holds.}

 When evaluating the right side in Eq. \eqref{eq-vNeum-quasifree} explicitly 
for monomials $\opA\in\mathcal A^{(2)}$, with 
\[\mathcal A^{(2)}:=\{\psi(x),\psi^*(x)\psi(y),\psi(x)\psi(y)\},\] 
one arrives at a system of coupled nonlinear
PDE's for  $(\phi_t,\gamma_t,\sigma_t)$, the \textit{Hartree-\-Fock\-Bogoliubov (HFB) equations}. Since quasifree states are uniquely determined by their truncated expectations $(\phi$, $\gamma$, $\sigma)$, the HFB equations are equivalent to equation \eqref{eq-vNeum-quasifree}. They are stated explicitly in 
\eqref{eq:BHF-phi} -~\eqref{eq:BHF-sigma}, below.

\begin{remark}   For states of systems of finitely many particles, such as gases used in BEC experiments in traps,  $\phi_t$ is square-integrable and $\gamma_t$ is a trace-class operator on $L^2(\R^d)$.  To study  states of systems in an infinite volume with an infinite number of particles and finite particle density, 
 one first replaces the one-particle space $L^2(\R^d)$ by $L^2(\Lambda)$, 
where $\Lambda$ is a compact d-dimensional set, for example $\Lambda := \T^d_L = \mathbb{R}^d/(L \Z)^d$ (a compact torus), and then one would pass to the thermodynamic limit,  $\Lambda \nearrow \mathbb{R}^{d}$.
\end{remark} 

To state our results we must introduce appropriate function spaces. Let $M:=\langle \nabla_{x}\rangle=\sqrt{1-\Delta_{x}}$,  where $\Delta_{x}$ is the $d$-dimensional Laplacian. 
 We denote by $\mathcal L^p$ the Schatten class of bounded operators, $A$, on  $L^2(\mathbb R ^d)$ with the property that $\tr |A|^p<\infty$ 
 and with the  norm $\|A\|_{\mathcal L^p}:=(\tr |A|^p)^{1/p}$. 
 For $j\in \mathbb N_0$, we define the spaces 
\begin{align}\label{eq-Xjspace-def-1}
X^{j} & =\big\{(\phi,\gamma,\sigma)\in 
H^{j} \times \mathcal{H}^{j}_\g \times \cH^{j}_\s: \gamma=\gamma^*\ge 0\ \text{ and }\ \sigma^*=\bar\sigma \big\}\,,
\end{align}
where $H^{j}$ are the standard Sobolev spaces $H^{j}(\mathbb{R}^{d})=M^{-j} L^{2}(\mathbb{R}^{d})$, $\mathcal{H}^{j}_\g=M^{-j}\mathcal{L}^{p} M^{-j}$ and  
 $\mathcal{H}^{j}_\s:=\{\s\in \mathcal L^2: \|M^{j}\s\|_{\mathcal L^2} + \|\s M^{j}\|_{\mathcal L^2} < \infty\}$. The norm on $X^{j}$ is defined as
\begin{align}\label{eq-Xjpqspace-def-2}
\|(\phi,\gamma,\sigma)\|_{X^{j}}=\|M^{j}\phi\|_{L^{2}} + \|M^{j}\gamma M^{j}\|_{\mathcal{L}^{1}}+ \|M^{j}\s\|_{\mathcal L^2} + \|\s M^{j}\|_{\mathcal L^2}. 
\end{align}

We will use the notation $\cX_T :=C^{0}([0,T); X^{j })\cap C^{3}([0,T); X^{0 })$, and we will denote by $X^{j}_{\rm qf}$ and $\cX_{T}^{\rm qf}$ 
the spaces of quasifree states and families of quasifree states, respectively, with $1^{st}$- and $2^{nd}$-order truncated expectations belonging to the spaces $X^{j}$ and $\cX_T$, respectively.

\begin{theorem}\label{prop-HFB-vfull-1}
Assume conditions (i) and (ii) of Sect. \ref{sec:many-body}.\,
Then $\qf_t\in \cX_{T}^{\rm qf} 
$  satisfies 
\eqn\label{eq-omt-commut-2}
    i\partial_t\qf_t(\opA) = \qf_t([\opA,\opH])\,,
    \qquad
    \forall \; \opA  \in \mathcal{A}^{(2)} \,, \; 
\eeqn
for the Hamiltonian $\opH$ given in Eq. \eqref{model-ham}, if and only if the triple $(\phi_t,\gamma_t,\sigma_t) = \mu(\qf_t)\in \cX_T 
$ of  $1^{st}$- and $2^{nd}$-order truncated expectations of $\qf_t$ satisfies the time-dependent
Hartree-Fock-Bogoliubov equations 
\begin{align}
i\partial_{t}\phi_{t} & =h (\gamma_t)\phi_{t}+k (\sigma^{\phi_{t}}_t)\bar{\phi}_{t} \label{eq:BHF-phi}\,,\\
i\partial_{t}\gamma_{t} & =[h (\gamma_t^{\phi_{t}}),\gamma_{t}]+k (\sigma^{\phi_{t}}_t)\sigma_{t}^*-\sigma_t k (\sigma^{\phi_{t}}_t)^* \label{eq:BHF-gamma}\,,\\
i\partial_{t}\sigma_{t} & =[h (\gamma_t^{\phi_{t}}),\sigma_{t}]_{+}+[k (\sigma^{\phi_{t}}_t),\gamma_{t}]_+ +k (\sigma^{\phi_{t}}_t), \label{eq:BHF-sigma}
\end{align}
where $[A_1,A_2]_{+}=A_1A_2^{T}+A_2A_1^{T}$, $\gamma^{\phi}:=\gamma+|\phi\rangle\langle\phi|$ and $\sigma^{\phi}:=\sigma +|\phi\rangle\langle\bar{\phi}|$, and 
\begin{align}\label{h}
& h (\gamma)  =h+b [\gamma]\,,\  b [\gamma]:= v*d(\gamma)  + v \,\sharp\, \gamma \,, \\
 \label{k}
 & k (\sigma)  =v\,\sharp \, \sigma\,,\ \quad d(\al)(x):=\al(x, x).
\end{align}
In these equations the operator $k : \al \ra v \,\sharp\, \al$ is defined through
\eqn  \label{v-sharp}
v \,\sharp\, \al \, (x;y):= v(x-y)\al (x;y) \,.
\eeqn
\end{theorem}
%

It is shown in \cite{BachBreteauxChenFroehlichSigal2016v2} that, for all times $t>0$, the r.h.s. of \eqref{eq:BHF-phi} - \eqref{eq:BHF-sigma} determine an element in the space $X^{0}$. 
%
%


\DETAILS{Self-consistent equations:  $(\phi_t,\gamma_t,\sigma_t)\in \cX_T$ 
satisfy the HFB equations \eqref{eq:BHF-phi} to~\eqref{eq:BHF-sigma} if and only if 
the corresponding  quasifree state $\qf_t\in \cX_{T}^{ \rm qf}$ 
  satisfies the self-consistent equation 
\eqn\label{eq-omt-commut-3}
    i\partial_t \qf_t(\opA) = \qf_t ({[\opA, \HHFB(\qf_t)]}) \,,
\eeqn
where $\HHFB(\qf), \qf\in X^{\rm qf},$  is the self-adjoint (see e.g. \cite{Derezinski2017}), quadratic Hamiltonian given for $(\phi,\gamma,\sigma)=\mu(\qf)\in X^{1}$ by} 
Moreover, the {\it quadratic HFB Hamiltonian},  $\HHFB(\qf), \qf\in X^{\rm qf},$ in the self-consistent equation \eqref{eq-qf-self-consist} is given by 
\begin{align}\label{HHFB-def}
    \HHFB(\qf) 
    & =\int \psi^{*}(x)h_v(\gamma)\psi(x)\, dx 
      \nonumber \\
    & \quad -\int b [|\phi \rangle\langle\phi|]\phi (x)\psi^{*}(x) \, dx+h.c.  
    \nonumber \\
    & \quad+\frac{1}{2}\int\psi^{*}(x)(v\# \s)\psi^{*}(x)\, dx+h.c. \,,
\end{align}
where $(\phi,\gamma,\sigma)$ are
 the $1^{st}$- and $2^{nd}$-order truncated expectations in the state $\omega^{q}$. The operator  $\HHFB(\qf), \, \qf\in X^{\rm qf},$ is a self-adjoint; see e.g.\ \cite{Derezinski2017}. 


For the pair potential $v(x, y)=g \delta(x-y)$, the HFB equations in a somewhat different form have first appeared in the physics literature; see \cite{DoddEdwardsClarkBurnett1998, Griffin1996, ParkinsWalls1981} and, for further discussion, \cite{BachBreteauxChenFroehlichSigal2016v2}.
 Here are some key properties of \eqref{eq:BHF-phi} - \eqref{eq:BHF-sigma} at a glance: 

\begin{itemize}
\item[(A)]    {\it Conservation of the total particle number}: If $\qf_t\in \cX_T^{\rm qf}$ solves Eq. \eqref{eq-omt-commut-2} (or \eqref{eq-qf-self-consist}) then the number of particles, 
\eqn 
    \cN(\phi_t,\gamma_t,\sigma_t) \, := \, \qf_t(\opN) \,, 
\eeqn 
where $\opN$ is the particle-number operator,  is conserved. 

\item[(B)]  {\it Existence and conservation of the  energy}: If $\qf_t\in  X^{\rm qf}_T$ solves \eqref{eq-omt-commut-2} then the energy 
\eqn 
    \mathcal{E}(\mu(\qf_t)) 
      :=\qf_t(\opH) 
\eeqn
 is conserved. Moreover, $\cE$ is given explicitly by the expression 
\begin{align}
    \mathcal{E}(\phi,\gamma,\sigma)=\tr[h(\gamma+|\phi\rangle\langle\phi|) &+b[|\phi\rangle\langle\phi|]\gamma
    +\frac{1}{2} b[\gamma]\gamma]\notag\\
 &   +\frac{1}{2}\int v(x-y)|\sigma(x,y)  +\phi(x)\phi(y)|^{2}dxdy\,.
    \label{energy}
\end{align}  

\item[(C)]  {\it Positivity preservation property}:  If $\Gamma=\big(\begin{smallmatrix} %
\gamma & \sigma\\
 \bar{\sigma} & 1+\bar \gamma
\end{smallmatrix}\big)\geq 0$ at $t=0$, then this holds for all times.


\item[(D)]  {\it Global well-posedness of the HFB equations}: See Theorem~\ref{thm:existence-uniquenes-up-to-Coulomb} below.

 \end{itemize}
Note that conservation of the  total particle number is related to invariance of the Hamiltonian $\mathbb{H}$ under the transformation $\psi^\sharp\rightarrow (e^{i\theta}\psi)^\sharp$, i.e., to $U(1)$-gauge invariance of the dynamics.
 
It is easy to verify that, under our assumptions, the operator in \eqref{HHFB-def} and the energy functional in \eqref{energy} are well defined. For example, for $(\phi,\gamma,\sigma)\in X^{1}$, we have that $\|v\# \s\|_{L^{2}}\ls \|M_{x-y} \s\|_{L^{2}}\ls  \|(M_{x} + M_{y} )\s\|_{L^{2}}\simeq \|\s\|_{\cH^{1}_\s}$, where $v\# \s, M_{x-y} \s$ and $(M_{x} + M_{y} )\s$ are treated as functions (integral kernels) in $L^{2}(\R^d_x \times \R^d_y)$.  
 
Statements (A) and (B) follow from the following general, yet elementary result.

\begin{theorem}
\label{thm:preservation-particle-number}
Let $\opA\in \cA^{(2)}$ be an operator commuting with the Hamiltonian, i.e., $[\opH,\opA]=0$. 
 Then $\qf_t(\opA)$ is   conserved:
 \eqn 
    \qf_t(\opA) = \qf_0(\opA) \;\;\;\;\forall\;t\in\R\,.
\eeqn
\end{theorem}
\begin{proof} 
This follows from \eqref{eq-omt-commut-2} for any operator $\opA$ quadratic in creation- and annihilation operators, with 
$[\opA,\opH]=0$. 
\end{proof} 

It is in the proof of the part of Statement (D) concerning local existence that an error was made in  \cite{BachBreteauxChenFroehlichSigal2016v1}. 
In the next two sections we look into this problem more closely.

In  \cite{BenedikterSokSolovej2018}, the program outlined in this paper has been pursued for equations analogous to the HFB equations valid for fermions, namely the Bogolubov-de Gennes equations; see also \cite{ChenSigal2017}. For references to related work see  \cite{BachBreteauxChenFroehlichSigal2016v2, ChenSigal2017, BenedikterSokSolovej2018}.

\begin{remark}
The HFB equations for $\phi_{t}$,
$\gamma_t$ and $\sigma_t$ stated in Theorem \ref{prop-HFB-vfull-1}
 can be reformulated in terms of $\phi_t$ and the generalized one-particle density matrix 
$\Gamma_t=\big(\begin{smallmatrix}\gamma_t & \sigma_t\\
\overline{\sigma}_{t} & 1+\overline{\gamma}_{t} \end{smallmatrix}\big)$. 
It has been shown in \cite{BachBreteauxChenFroehlichSigal2016v2} that the  diagonalizing maps, $\mathcal{U}_t$, for $\Gamma_t$ are ``symplectomorphisms'' and that the equation of motion for $\Gamma_t$ is equivalent to an evolution equation for these symplectomorphisms. The latter property allows us
\begin{enumerate}
 \item[(a)] {to give another proof of the conservation of energy without using the second quantization formalism;}
 \item[(b)] {to express the energy functional $\mathcal{E}$ in terms of the diagonalizing maps $\mathcal{U}_t$ and $\phi_t$ and to interpret it as a Hamilton functional on an infinite-dimensional, affine, complex phase space;}
\item[(c)] {to show that the HFB equations are equivalent to the Hamiltonian equations of motion for 
$(\mathcal{U}_t, \phi_t)$; and}
 \item[(d)] {to relate the time-dependent HFB equations \eqref{eq:BHF-gamma} - \eqref{eq:BHF-sigma} to the  time-independent HFB equations used in the physics literature.}
\end{enumerate}
\end{remark}


\section{Existence and Uniqueness of Solutions} 
\label{sec-WP-HFB}


We recall that, given a Banach space $X$, a function $f\in C(X)$ continuous on $X$, and the infinitesimal generator $-iA$ of a strongly continuous semigroup $G(t)$ on $X$, 
 a continuous function $\rho:[0,T)\to X$ is called a \textit{mild solution} of the equation
 \begin{equation}\label{eq:def_problem}
i\partial_t \rho_t =A\rho_t +f(\rho_t)\,, \qquad \rho_{t=0} =\rho_{0}\in X\,,
\end{equation}
iff $\rho_t$ solves the fixed point equation in integral form  (with the integral understood in the sense of Bochner)\begin{equation}\label{eq:def_mild_HFB}
\rho_t=G(t)\rho_0 -i\int_0^t G(t-s) f(\rho_s)\,ds.
\end{equation}

We have the following result
\begin{theorem}
\label{thm:existence-uniquenes-up-to-Coulomb}
Suppose that $d\leq3$, and let $\rho_{0}=(\phi_{0},\gamma_{0},\sigma_{0})\in \XL$. 
Suppose, furthermore, that the potentials $V$ and $v$ satisfy conditions (i) and (ii')
 of Section \ref{sec:many-body}. 
Then the following hold:
\begin{enumerate}
\item[(a)] \label{ite:local-existence-mild}
\textit{Existence and uniqueness of a local mild solution:}
There exists some $T$, with $0<T\leq \infty$, such that the HBF equations \eqref{eq:BHF-phi}-\eqref{eq:BHF-sigma} 
have a unique maximal mild  solution
\[
(\rho_{t})_{t\in[0,T)}=(\phi_{t},\gamma_{t},\sigma_{t})_{t\in[0,T)}\in C^{0}([0,T);\XL).
\]

\item[(b)]\label{ite:local-existence-classical}
\textit{Existence and uniqueness of a local classical solution:}
If $\rho_{0}\in X^{3}$, then 
\[
(\rho_{t})_{t\in[0,T)}\in C^{0}([0,T);X^{3})\cap C^{1}([0,T);\XL)
\]
and $\rho_{t}$ satisfies the HBF equations \eqref{eq:BHF-phi}-\eqref{eq:BHF-sigma} in the classical sense.

\item[(c)]\label{ite:conservation-laws}
\textit{Conservation laws:}
The number of particles $\tr[\gamma_t]$ and the energy \eqref{energy} are constant in time.

\item[(d)]  {\it Positivity preservation property}:  If $\Gamma=\big(\begin{smallmatrix} 
\gamma & \sigma\\
 \bar{\sigma} & 1+\bar \gamma
\end{smallmatrix}\big) \geq 0$ at $t=0$, then this property holds for all times.

\item[(e)]\label{ite:global-existence}
\textit{Existence of a global solution:}
If additionally $\Gamma_0:=\big(\begin{smallmatrix}\gamma_0 & \sigma_0\\ \bar{\sigma}_0 & 1+\bar{\gamma}_0\end{smallmatrix}\big)\geq 0$, then
 the solution $\rho_t$ is global, i.e.,~$T=\infty$.
\end{enumerate}
\end{theorem}


%


In \cite{BachBreteauxChenFroehlichSigal2016v1}, we claimed this result under  conditions similar to (i) and (ii) of Section \ref{sec:many-body}, but a mistake infiltrated the proof of (a) (see Section \ref{sec-k-est}). Below, we sketch the main ideas of the proof of this theorem, and in Section \ref{sec-k-est} we give a correct proof of the estimate in question, but under the stronger conditions (i) and (ii'). 
 A proof under  conditions (i) and (ii) of Section \ref{sec:many-body} will be given elsewhere. 

First, setting $\rho:=(\phi,\gamma,\sigma)$ and separating the linear part, $A\rho$, from the non-linear part, $f(\rho)$, we can write the HFB equations \eqref{eq:BHF-phi} to \eqref{eq:BHF-sigma} in the form given in \eqref{eq:def_problem}, with
\eqn \label{eq:def-A}
A\rho=\big( h\phi \,,\, [h,\gamma] \,,\, [h,\sigma]_{+}+k[\sigma] \big)\,,
\eeqn
where the operators $h$, $k$ (and $b$ -- see below) are defined in \eqref{h} -  \eqref{v-sharp}, and with the non-linear part $f:=(f_{1},f_{2},f_{3})$ given by
\begin{align}
\label{eq:def_f_1}
f_{1}(\rho) & =b[\gamma]\phi+k[\sigma+\phi^{\otimes2}]\bar{\phi}\,, \\
\label{eq:def_f_2}
f_{2}(\rho) & =[b[\gamma+|\phi\rangle\langle\phi|],\gamma]+k[\sigma+\phi^{\otimes2}]\bar{\sigma}-\sigma\overline{k[\sigma+\phi^{\otimes2}]}\,, \\
\label{eq:def_f_3}
f_{3}(\rho) & =[b[\gamma+|\phi\rangle\langle\phi|],\sigma]_{+}+[k[\sigma+\phi^{\otimes2}],\gamma]_{+}\,.
\end{align}

 The proof of statement (a) is based on applying a standard fixed point argument to \eqref{eq:def_mild_HFB}, (using the 
Picard-Lindel\"of theorem).  To this end we show that $A$ generates 
 a strongly continuous, uniformly bounded semigroup on $\XL$, (in particular, that
the map $t\mapsto\|G(t)\|_{\cB(\XL)}$ is bounded), and that  $f$ 
 is locally Lipschitz.


To prove global existence, we use the fact that the kinetic energy operator
 \eqn \label{eqn:def-mathbb-T}
\mathbb T := \int dxdy \; \psi^*(x)(-\Delta)\psi(y) \,
\eeqn
controls, and is controlled, 
by the Hamiltonian operator $\opH$ introduced in \eqref{model-ham} 
and the number operator $\opN$. More precisely, the following inequalities hold in the sense of quadratic forms. 
\begin{align} \label{eq:control-T-by-H-and-N}
\frac{2}{3} \opH - C \opN^{2} 
\leq 
\mathbb T 
\leq 
\opH + C \opN^{2} \, ,
\end{align}
where $C \equiv C_{V,v} < \infty$ depends only on the external
potential $V$ and the pair potential $v$.


 Taking expectations of all the terms in \eqref{eq:control-T-by-H-and-N}
in the state $\qf_t$ we observe that, for an arbitrary positive integer $k$, there
exists a universal constant $C_k < \infty$ such that
\begin{align} \label{eq:control-N-qf}
\qf_t(\opN)^{k} 
\leq 
\qf_t\big( \opN^{k} \big) 
\leq 
\big[ \qf_t(\opN) + C_k \big]^{k} \, .  
\end{align}
The first inequality in \eqref{eq:control-N-qf} follows from the Jensen inequality while for the second one we have used that $\qf_t$ is quasifree.  Hence, using conservation 
of the particle number $\qf_t(\opN) = \qf_0(\opN)$ and of the energy 
$\qf_t(\opH) = \qf_0(\opH)$, we obtain upper and lower bounds on
$\qf_t(\mathbb T)$ in terms of $\qf_0(\mathbb T)$, \textit{uniformly} in $t$.
These bounds then imply
\DETAILS{by the Schr\"odinger operato $\opH$, \eqref{model-ham}, and the number of particles, $\opN$, as 
\begin{align}
\mathbb T \leq 3\opH + C \opN^{2} \,.
\end{align}
We now apply to this $\qf_t$ and use the conservation of the particle number and of the energy to obtain}
bounds on $\|\gamma_t\|_{\cH^{1}_\g}$ and $\|\phi_t\|_{\Hd}$ that are uniform in $t$. 
Moreover, uniform bounds on $\|\sigma_t\|_{\cH^{1}_\s}$ are obtained from the estimate (\cite{BachBreteauxChenFroehlichSigal2016v2}) 
 \begin{equation} \label{sig-est}\|\sigma\|^2_{\cH^{1}_\s}\leq 2\|\gamma\|_{\cH^{1}_\g}(1+ \tr[\gamma])\,,\end{equation} 
  which follows from the definitions in \eqref{omq-mom}, see  inequality \eqref{Gam-pos}. 
We thus conclude that the solution is global.

\begin{remark}    One can reformulate \eqref{eq-qf-self-consist} as a fixed point problem (see \cite{BachBreteauxChenFroehlichSigal2016v2}) which suggests a possibility of proving the existence result directly for \eqref{eq-qf-self-consist} without going to  the truncated expectations of $\qf$.
\end{remark}

\section{Estimate of the operator $k$} 
\label{sec-k-est}

    As was mentioned above, one of the key steps in the fixed-point argument used in the proof of statement (a) is to show that $f$  is locally Lipschitz. 
Here one uses the estimate
 \begin{align}\label{Mk-est}
\|Mk[\sigma]\|_{\cL^{2}} &\ls \|\s\|_{\cH^{1}_\s},\end{align}  
which implies the necessary estimates on the terms $k[\sigma]\bar{\phi}, $ $k[\sigma]\bar{\sigma}$, $\sigma\overline{k[\sigma]}$, $k[\sigma]\bar{\gamma}$ and $\gamma k[\sigma]$
   in \mbox{Eq.~\eqref{eq:def_f_1}$-$\eqref{eq:def_f_3}.} It is exactly in the proof of \eqref{Mk-est} - see Lemma E.1(2) of \cite{BachBreteauxChenFroehlichSigal2016v1} - 
    where an error has occurred in \cite{BachBreteauxChenFroehlichSigal2016v1}. 
   
Here we prove \eqref{Mk-est} under the condition 
 that $v \in W^{p, 1}$ with $p>d$. 

In what follows we use the notation  $A\lesssim B$ to represent an inequality of the form $A\leq c B$, for some positive constant $c$. 
 We begin our proof with the following lemma.

 \begin{lemma} \label{lem:k-est} 
Assume that $v \in W^{p, 1}$, with $p>d$. 
Then the operator 
$k$ defined in \eqref{k} and \eqref{v-sharp}
satisfies the bound \eqref{Mk-est}.
\DETAILS{\begin{itemize}
\item \label{enu:continuity-K-H1L2-L2H1} $k$ is 
continuous from $\Hsdd$ to $M^{-1}\cL^2(L^2)$. 
\end{itemize}}
\end{lemma}

\begin{proof} 
Denote   by $\tilde\s$ 
 the (generalized) integral kernel of an operator $\s$. 
Clearly, $\|\sigma\|_{\mathcal{H}^j_\s}\simeq \|\tilde\s\|_{H^{1}}$. 
Denote by $a(x,y)=v(x,y) \tilde\sigma(x, y)$, 
 the integral kernel of $k[\sigma]$. We have that
\begin{align}\label{Mk-est2}
\|Mk\|_{\cL^{2}}^{2} & =\int\int|M_{x}a(x,y)|^2dxdy \le\|a\|_{H^{1}}^{2}.\end{align}
Since $a(x,y)=v(x,y) \tilde\s(x,y)$  and 
\[\|a \|_{H^{1}}\le \|a \|_{L^{2}}+ \|\p_x a \|_{L^{2}}+ \|\p_y a \|_{L^{2}}\,,\] 
we use the Leibniz rule,  $\p_x a(x,y)=(\p_x v(x,y)) \tilde\s(x,y) + v(x,y) \p_x \tilde\s(x,y)$, to find that
\begin{align}
\label{a-estim}\|a \|_{H^{1}}\le &(\|v \|_{L^{\infty}} 
+ \|\p_x v M_x^{-1}\|+ \|\p_y v M_y^{-1}\|) \|\tilde\s \|_{H^{1}},\end{align}
where the norms without subindices are the operator norms for operators on $L^2(\R^d_x\times \R^d_y)$. 
The Schwartz and Sobolev inequalities imply that $$\|\p_x v f\|_{L^{2}}\le  \|\p_x v \|_{L^{p}}\|  f\|_{L^{s}}\ls \|v \|_{W^{p, 1}}\| M f\|_{L^{2}},$$ for arbitrary $s$ and $p$ satisfying $\frac1p+\frac1s = \frac12$ and $p>d$. Thus $$\|\p_x v M_x^{-1}\|\ls \|v \|_{W^{p, 1}}\,,$$ and, similarly, $\|\p_y v M_y^{-1}\|\ls  \|v \|_{W^{p, 1}}$. It follows that
\begin{align}\label{a-est'}
\|a \|_{H^{1}}\ls \|v \|_{W^{p, 1}}\|\tilde\s\|_{H^{1}}.\end{align}
This,
 together with \eqref{Mk-est2} and $\|\tilde\s\|_{H^{1}}\simeq \|\sigma\|_{\mathcal{H}^1_\s}$, yields \eqref{Mk-est}. 
\end{proof}
\begin{corollary} The following estimates hold true 
\begin{align}\label{k-est}
\|k[\sigma]\bar{\phi}\|_{\Hd}
\ls \|\sigma\|_{\mathcal{H}^1_\s}\|\phi\|_{L^2},\ \|k[\sigma]\bar{\sigma}\|_{\mathcal{H}^1_\g} 
\ls \|\sigma\|_{\mathcal{H}^1_\s}^{2} \,, \|k[\sigma]\bar{\gamma}\|_{\mathcal{H}^1_\s}   \ls \|\sigma\|_{\mathcal{H}^1_\s}\|\gamma\|_{\mathcal{H}^1_\g},  \end{align} 
 and similarly for the terms  $\sigma\overline{k[\sigma]}$ and $\gamma k[\sigma]$. \end{corollary}
\begin{proof} 

 For~$k[\sigma]\bar{\phi}$, we use Lemma~\ref{lem:k-est} 
 to find that
$$
\|k[\sigma]\bar{\phi}\|_{\Hd}
\leq\|M k[\sigma]\|_{\cB}\|\bar{\phi}\|_{L^2}\leq C\|\sigma\|_{\mathcal{H}^1_\s}\|\phi\|_{L^2}\,.
$$
For $k[\sigma]\bar{\sigma}$ (and, similarly, for $\sigma\overline{k[\sigma]}$), 
the inequality $$\|k[\sigma]\bar\s\|_{\cH^{1}_\g}=\|Mk[\sigma]\bar\s M\|_{\cL^{1}}\le\|M k[\sigma]\|_{\cL^{2}}\|\bar\s M\|_{\cL^{2}}$$ 
and Lemma~\ref{lem:k-est} 
 (see estimate \eqref{Mk-est}),  
as well as the bound $\|\bar\s M\|_{\cL^{2}}\le \|\s\|_{\mathcal{H}^1_\s}$ (which follows from the definition of $\|\s \|_{\mathcal{H}^1_\s}$) yield the second estimate in \eqref{k-est}.

%
\DETAILS{For $b[\gamma]\sigma$ (or similarly $\sigma\overline{b[\gamma]}$), using Lemma~\ref{prop:continuity_of_B_and_K}\eqref{enu:continuity-B-H1L2-LH1}, we obtain
\begin{align*}
\|b[\gamma]\sigma\|_{\Hsdd}  \leq\|M b[\gamma]M^{-1}\|_{\cB(L^{2})}\|M\sigma M\|_{\mathcal{L}^{2}(L^{2})}
  \leq C\|\gamma\|_{\cHLd}\|\sigma\|_{\Hsdd} \,.
\end{align*}}
To conclude we note that, for $k[\sigma]\bar{\gamma}$ (and, similarly, for $\gamma k[\sigma]$), using Lemma~\ref{lem:k-est} 
 (see estimate \eqref{Mk-est}), we arrive at the inequality
\begin{align*}
\|k[\sigma]\bar{\gamma}\|_{\mathcal{H}^1_\s}  \leq\|M k[\sigma]\|_{\mathcal{L}^{2}}\|\bar{\gamma}\|_{\cB} + \| k[\sigma]\|_{\mathcal{L}^{2}}\|\bar{\gamma}M\|_{\cB}
  \leq C\|\sigma\|_{\mathcal{H}^1_\s}\|\gamma\|_{\mathcal{H}^1_\g}\,,
\end{align*}
where, recall, $\cB$ is the space of bounded operators on $L^{2}(\R^d)$. This completes the proof.  
  \end{proof}

\subsection*{Acknowledgements}
The authors are grateful to J.~Sok for pointing out an error in
\cite{BachBreteauxChenFroehlichSigal2016v1}. The work of S.B.\ has
been supported by the Basque Government through the BERC 2014-2017
program, and by the Spanish Ministry of Economy and Competitiveness
MINECO (BCAM Severo Ochoa accreditation SEV-2013-0323, MTM2014-53850),
and the European Union's Horizon 2020 research and innovation
programme under the Marie Sklodowska-Curie grant agreement No 660021.
The work of T.C.\ is supported by NSF grants DMS-1151414 (CAREER) and
DMS-1716198. The work of I.M.S.\ is supported in part by NSERC Grant
No.~NA7901 and by SwissMAP.

\bibliographystyle{plain}

\end{document}